\documentclass{article}

\author{Umberto Cherubini \and Fabio Gobbi  \and Sabrina
Mulinacci \and Silvia Romagnoli
\\ University of Bologna - Department of Statistics}

\usepackage{amsmath}
\usepackage{amsthm}
\usepackage{amsfonts}
\usepackage{amssymb}
\usepackage{latexsym}
\usepackage[dvips]{graphicx}
\usepackage{pstricks}
\usepackage{epsfig}

\newtheorem{theorem}{Theorem}[section]

\newtheorem{lemma}{Lemma}[section]
\newtheorem{proposition}{Proposition}[section]
\newtheorem{definition}[theorem]{Definition}
\newtheorem{example}{Example}[section]

\newtheorem{remark}{Remark}[section]

\begin{document}

\title{Granger Independent Martingale Processes}
\date{}

\maketitle

\begin{abstract} We introduce a new class of processes for the evaluation of multivariate equity derivatives. The proposed setting is well suited for the application of the standard copula function theory to processes, rather than variables, and easily enables to enforce the martingale pricing requirement. The martingale condition is imposed in a general multidimensional Markov setting to which we only add the restriction of no-Granger-causality of the increments (Granger-independent increments).
We call this class of processes GIMP (Granger Independent
Martingale Processes). The approach can also be extended to the
application of time change, under which the martingale restriction
continues to hold. Moreover, we show that the class of GIMP
processes is closed under time changing: if a Granger independent
process is used as a multivariate stochastic clock for the change
of time of a GIMP process, the new process is also GIMP. \end{abstract}

\bigskip

{\bf Keywords}: Option pricing, Granger causality, copula
functions, Garch processes
\medskip

{\bf Mathematics Subject Classification (2010)}: 60G44 60J25 91G20
\medskip

{\bf JEL Classification}: G13  G11

\section{Introduction}

Multivariate equity options are largely used both in structured
finance and index-linked life insurance policies. Their payoffs
depend on the application of an aggregation function to a set of
underlying asset prices. Examples are altiplano notes that promise
a payoff if all assets are above a given threshold, Everest notes
that use the $\min (X_1,X_2,\ldots , X_m)$ as the aggregation
function, or basket options that use arithmetic average as
aggregation function.
Copula functions have been widely applied to the evaluation of these products. In the typical application, both in the industry and the literature, marginal distributions are calibrated on univariate option prices, and the copula is estimated from the time series of the underlying assets. The first proposals of these models go back to Rosenberg (1999) and Cherubini and Luciano (2002). Rosenberg(2003) provided a risk neutral pricing model in a static setting. Finally, Van den Goorbergh et al. (2005) extended the model allowing for time varying dependence. Their approach is very close to ours, since they extend the no-arbitrage assumption to a dynamic setting by imposing a condition on the information structure that we call no-Granger causality in this paper.\\
The main contribution of this paper is to provide a systematic
analysis of the properties of the class of processes used in Van
den Goorbergh et al. (2005). We call these processes
Granger-Independent Martingale Processes (GIMP). We will show that
the time change technique can be applied to this class of
processes, and that Granger independent processes are endowed with
the same closure property as L\'evy processes. Namely, if one uses
a Granger-independent process as a stochastic clock for the time
change of a Granger-independent process, he produces a process
that is part of the same class. Moreover, the time change
technique enables to extend the applicability of the model to a
wider set of processes.
One can in fact generate a Granger independent martingale process and apply an arbitrary stochastic clock to it, so that the martingale property is preserved even though the stochastic clock is not 
part of the class of Granger independent processes.

The paper is organized as follows. In section \ref{dynamics} we
present a general multivariate arbitrage free model in the spirit
of coupling marginal arbitrage free price processes. In section
\ref{tc} we will address the issue of time change, proving that
the class of GIMP is invariant under change of time. Section
\ref{conclusione} concludes.

\bigskip

\section{A multivariate model for price dynamics}\label{dynamics}

We now describe the model for the market price dynamics of asset
returns. Let $\left\{ \mathbf{\Omega },\mathcal F, \left (\mathcal
F _{t}\right )_{t\in \mathbb N},\mathbb P \right\} $ be the
underlying filtered probability space. The setting is
multivariate, so that we denote with $X_t^j$, $j=1,2,...,m$, the
log-prices of assets in the economy at time $t$ and with
$\mathbf{X}=(X^1,...,X^m)$ the multidimensional stochastic
process. Correspondly, we denote with $\mathcal (F_{t}^j)_{t\in
\mathbb N}$ the filtration collecting the information generated by
the history of asset $j$ up to time $t$ and with $\left (\mathcal
F^X_t\right )_{t\in \mathbb N}$ the filtration generated by the
multidimensional stochastic process $\mathbf{X}
$.\\

\subsection{Granger independent processes}\label{multimartingale}

In the arbitrage-free pricing theory, it is required that the
underlying discounted prices are martingale with respect to the
filtration containing information generated by all the assets in
the basket. In models based on the specification of the joint
distribution (such as the multivariate generalization of the Black
and Scholes model in the seminal Johnson, 1987 and Margrabe, 1978
papers) the martingale property is included from the very start,
at the cost of making calibration of the marginal distributions
more difficult. In copula-based models, the martingale condition
should be impounded in the model once the martingale marginal
processes have been specified. We are going to show that this is
linked to a hypothesis that is very well known in econometrics,
and is called no-Granger causality. No-Granger-causality means
that no information can help to predict the future values of a
variable over and above its past history.

\begin{definition}\label{granger}
$X^1,\ldots, X^{i-1}, X^{i+1},\ldots , X^m$ do not Granger cause
$X^i$ if and only if
$$
\mathbb P[X^i_{t+1}\leq x\vert\mathcal F^X_{t}]=\mathbb
P[X^i_{t+1}\leq x\vert\mathcal F^i_{t}]
$$
for any $t\in \mathbb N$ and $x$, or, equivalently, in terms of
increments if and only if
\begin{equation}\label{inc1}
 \mathbb P(\Delta X^k_{t}\leq x\vert \mathcal F^X_t)= \mathbb P(\Delta X^k_{t}\leq x\vert \mathcal F^k_t),
\end{equation}
where $\Delta X^k_{t}=X^k_{t+1}-X^k_t$, for any $t\in \mathbb N$
and $x$.
\end{definition}
The absence of Granger causality induces that, if the martingale
property or the Markov property are satisfied by each process with
respect to its own filtration, then they also hold with respect to
the filtration generated by the whole multidimensional process. As
for the martingale property, this is obvious since the no-Granger
causality says that the distributions of $X^i_{t+1}$ conditioned
on $\mathcal F^i_{t}$ and $\mathcal F^X_{t}$ are the same.

Intuitively, (\ref{inc1}) implies that, while the marginal
conditional distribution of increments only depends on the
corresponding marginal filtration, the conditional dependence
structure of the vector of increments could depend on the whole
filtration $\left (\mathcal F^X_t\right )$.\\
The following is a particular specification of (\ref{inc1}).

%Since the dependence among variables is determined by copulas, we
%are then naturally led to use the concept of \emph{conditional
%copula} introduced in Patton (2006), that we slightly modify
%conditioning on the filtration.
%\begin{definition}
%Let $\mathbf{Z}=(Z^1,\ldots ,Z^m)$ be an $\mathbb R^m$-valued
%random vector and $\mathcal A$ a $\sigma$-algebra. If, for
%$j=1,\ldots ,m$, $U^j=F^j(Z^j\vert \mathcal A)$ (where
%$F^j(\cdot\vert \mathcal A)$ is the distribution function of $Z^j$
%given $\mathcal A$), then the conditional copula of $\mathbf{Z}$
%given $\mathcal A$ is the conditional joint distribution function
%of the random vector $\mathbf{U}=(U^1,\ldots ,U^m)$ given
%$\mathcal A$.
%\end{definition}

\begin{definition} We say that a multidimensional stochastic process $\mathbf{X}=(X^1,\ldots ,X^m)$
has \emph{Granger} independent increments if for any $k=1,\ldots
,m$
$$\mathbb P(\Delta X^k_{t}\leq x\vert \mathcal F^X_t)=
\mathbb P(\Delta X^k_{t}\leq x),$$ where $\Delta
X^k_{t}=X^k_{t+1}-X^k_t$, $t\in \mathbb N$.
\end{definition}

This notion of independence is in the middle between
\emph{componentwise independence} of increments
$$\mathbb P(\Delta X^k_{t}\leq x\vert \mathcal F^k_t)=
\mathbb P(\Delta X^k_{t}\leq x),$$ and \emph{vector independence}
of increments
$$\mathbb P(\Delta X^1_{t}\leq x_1,\ldots ,\Delta X^m_{t}\leq x_m\vert \mathcal F^X_t)
= \mathbb P(\Delta X^1_{t}\leq x_1,\ldots ,\Delta X^m_{t}\leq
x_m).$$
\medskip

While vector independence requires that the \emph{joint}
distribution of increments be independent of the levels, Granger
independence only requires that the \emph{marginal} distribution
of the increments be independent of them.

Thanks to Sklar's theorem (see Nelsen,2006), it is now immediate
to convince oneself that if the $\mathbb R^m$-valued stochastic
process $\mathbf{X}$ has Granger independent increments and the
conditional copula of the vector of increments $(\Delta
X^1_{t},\ldots ,\Delta X^m_{t})$ given $\mathcal F^X_t$ is
independent of $\mathcal F^X_t$, then the process exhibits vector
independent increments. Moreover, it is likewise obvious that, if
the conditional copula of the vector of increments $(\Delta
X^1_{t},\ldots ,\Delta X^m_{t})$ given $\mathcal F^X_t$ coincides
with the copula of the vector of increments $(\Delta
X^1_{t},\ldots ,\Delta X^m_{t})$ conditioned on
$\mathbf{X}_t=(X^1_t,\ldots , X^m_t)$, than the stochastic process
$\mathbf{X}$ is a multidimensional Markov process.
\medskip

Multidimensional Granger independent processes are very well
suited to construct multidimensional exponential martingales, that
is to model assets prices with respect to the risk neutral pricing
measure. In fact, each unidimensional exponential process, being
endowed with independent log-increments, satisfies the martingale
property with respect to its natural filtration, once it is
normalized with respect to its mean. Now, thanks to no-Granger
independence, the process automatically turns out to be a
multidimensional martingale. This fact justifies the following
definition

\begin{definition} We define Granger-independent martingale process (GIMP) a
multivariate process with Granger independent log-increments in
which each univariate process is a martingale with respect to its
own natural filtration.
\end{definition}

The above approach has been already applied in literature in the case of multivariate Garch processes and an equivalent martingale probability introduced, see van der Goorberg et al. (2005). Here below we analyze the same 
case showing that the
 equivalent martingale change of probability preserves the no-Granger causality.

\begin{example}\label{garchex}
\emph{\textbf{Copulas for GARCH martingale processes.}} In this
example we analyze a specific model in which the
increments of each process follow a GARCH dynamics. \\
More precisely, let $t$ from $0$ to $N>0$ and denote with $\Omega
^j$ the set of all possible trajectories from time $0$ to time $N$
of the process $X^j$. Let $\Omega=\Omega
^1\times\cdots\times\Omega ^m$ be the set of all multidimensional
paths. $\mathcal F^j_t$ and $\mathcal F ^X$ denote, as above, the
filtration generated by the process $X^j$ and by the whole market
$\mathbf{X}$, respectively. We consider the probability space
$\left (\Omega,\mathcal F^X_N,\mathbb P\right )$ where $\mathbb P$
denotes the objective probability and we call $\mathbb P^j$ the
projection of $\mathbb P$ on $\left (\Omega ^j,\mathcal
F^j_N\right )$. As for the dynamics of the processes, we assume
that the increment $Y^j_t$ of each $X^j$ from time $t-1$ to time
$t$ follows, with respect to $\mathbb P^j$, a dynamics of type
$$\begin{aligned}
   &Y^j_t=\mu^j_t-\frac{(H^j_t)^2}2+H^j_tZ^j_t\\
&(H^j_t)^2=\omega^j _0+\omega^j_1(H^j_{t-1})^2+\omega^j _2(Y^j_{t-1})^2\\
&Z^j_{t}\sim N(0,1)\text{ i.i.d.}
  \end{aligned}
$$
with $\mu^j_t$ $\mathcal F^j_{t-1}$-adapted and $\omega^j_i$
positive constants for $i=0,1,2$. Moreover we assume that
$$\begin{aligned}
&\mathbb P\left (Y^1_t\leq y^1,\ldots ,Y^m_t\leq y^m\vert Y^1_1=z^1_1,\ldots ,Y^1_{t-1}=z^1_{t-1},\ldots ,Y^m_1=z^m_1,\ldots ,Y^m_{t-1}=z^m_{t-1}\right )=\\
&=C_{t\vert Y^1_1=z^1_1,\ldots ,Y^1_{t-1}=z^1_{t-1},\ldots
,Y^m_1=z^m_1,\ldots ,Y^m_{t-1}=z^m_{t-1}}(F_{t-1}^1(y^1),\ldots
,F_{t-1}^m(y^m))
  \end{aligned}$$
where $F_{t-1}^j(y^j)=\mathbb P^j(Y^j_t\leq y^j\vert
Y^j_1=z^j_1,\ldots ,Y^j_{t-1}=z^j_{t-1} )$ and where we suppose
that each conditional copula
$$C_{t\vert Y^1_1=z^1_1,\ldots ,Y^1_{t-1}=z^1_{t-1},\ldots ,Y^m_1=z^m_1,\ldots ,Y^m_{t-1}=z^m_{t-1}}(u_1,\ldots ,u_m)$$
has a strictly positive density in $(0,1)^m$.\\
This model clearly satisfies the no-Granger assumption.
\medskip

We now introduce the martingale restriction on the marginal
processes, following the general approach in Christoffersen,
Elkamhi and Feunou (CEF, 2010), for which we construct the
multivariate extension. For the sake of simplicity, we assume zero
risk-free rate. CEF (2010) show that for each marginal process $j$
there exists an equivalent probability $\mathbb Q^j\sim \mathbb
P^j$ with respect to which the increments of the log-prices
satisfy
$$\begin{aligned}
   &Y^j_t=-\frac{(H^j_t)^2}2+H^j_t\tilde Z^j_t\\
&(H^j_t)^2=\omega^j _0+\omega^j_1(H^j_{t-1})^2+\omega^j _2(Y^j_{t-1})^2\\
&\tilde Z^j_{t}\sim N(0,1)\text{ i.i.d.}
  \end{aligned}
$$
and, then $ S^j_t=e^{X^j_t}$ is an $\mathcal F^j_t$-martingale.\\

We now define the joint probability $\mathbb Q$ on $\left
(\Omega,\mathcal F^X_N\right )$  as
$$\mathbb Q\left (Y^1_t\leq y^1,\ldots ,Y^m_t\leq y^m\vert\mathcal F^X_{t-1}\right )=
C_{t\vert\mathcal F^X_{t-1}}(G_{t-1}^1(y^1),\ldots
,G_{t-1}^m(y^m))$$
where $G^j_{t-1}(y^j)=\mathbb Q^j(Y^j_t\leq y^j\vert Y^j_1=z^j_1,\ldots ,Y^j_{t-1}=z^j_{t-1})$.\\
$\mathbb Q$ is equivalent to $\mathbb P$ and the distribution with
respect to $\mathbb Q$ of the multidimensional process
$\mathbf{X}$ continues to satisfy the no-Granger causality
assumption. Since each $S^j_t$ is a $\mathcal F^j_t$-martingale
with respect to $\mathbb Q$, then it is a $\mathcal
F^X_t$-martingale as well.
\end{example}

\section{Time Changed Multidimensional Martingale Processes}\label{tc}
In this section we will consider the case of multidimensional
time-changed stochastic processes. We will develop the time change
approach in two directions. The first one aims at providing a
technique to build martingales with respect to the filtration
generated by the whole multidimensional process, given that they
satisfy the martingale property with respect to their natural
filtration. %More specifically, we will show that if under a given
%time clock each process is a martingale with respect to its
%natural filtration and the no-Granger causality assumption holds,
% then it remains a multidimensional martingale with
%respect to the filtration generated by the whole process after an
%independent multidimensional time change. This result enables to
%extend our approach of building multidimensional martingale
%processes from uni-dimensional ones,
%allowing for cases in which the processes can Granger cause each other through the business time activity, i.e. the stochastic clock.\\
The second direction gives a technique to construct new Granger
independent increments processes: this is obtained assuming that
the processes have stationary increments and that the stochastic
clock has Granger independent increments as well.
\medskip

In what follows, we will represent the multidimensional stochastic
clock as an $\mathbb R^m$-valued process
$\mathbf{T}=(\mathbf{T}_t)_{t\in \mathbb N }=(T_t^1,\ldots
,T_t^m)_{t\in \mathbb N }$ such that $\mathbf{T}_0=(0,\ldots ,0)$
and such that each component $T^j$ is an increasing
process.\smallskip

Given an $\mathbb R^m$-valued stochastic process $\mathbf{Y}=\left
(\mathbf{Y}_t\right )_{t\in\mathbb N}=(Y^1_t,\ldots
,Y^m_t)_{t\in\mathbb N}$ we denote with
$\mathbf{Y}_T=(Y^1_{T^1},\ldots ,Y^m_{T^m})$ the corresponding
time changed process whose components are the stochastic processes
$\left (Y^j_{T^j_t}\right )_{t\in \mathbb N }$.

\begin{lemma}\label{bye}
Let $\mathbf{X}=(X^1,\ldots ,X^m)$ be a multidimensional
stochastic process such that each one-dimensional stochastic
process is not Granger caused by the others and $\mathbf{T}$ be a
multidimensional stochastic clock. We assume $\mathbf{T}$
independent of $\mathbf{X}$. Then
\begin{equation*}
\mathbb P(X^j_{T^j_{s+1}}-X^j_{T^j_s}\leq x\vert\mathcal
F_s^{X_T,T})=\mathbb P(X^j_{T^j_{s+1}}-X^j_{T^j_s}\leq x \vert
\mathcal F^{X^j_{T^j},T}_s).
\end{equation*}
where $\left (\mathcal F_s^{X_T,T}\right )$ is the filtration
generated by the $2m$-dimensional stochastic process
$(\mathbf{X}_{T_s},\mathbf{T}_s)$ and $\left(\mathcal
F_s^{X^j_{T^j},T}\right )$ that generated by the $m+1$-dimensional
process $(X^j_{T^j_s}, \mathbf{T}_s)$.

\end{lemma}

\begin{proof} For $s\in\mathbb N$ we have
$$\begin{aligned}
&\mathbb P(X^j_{T^j_{s+1}}-X^j_{T^j_s}\leq x\vert \mathbf
{X}_{\mathbf{T}_r}=
{\bf x}_r,\mathbf{T}_r={\bf t}_r,\, r\in\mathbb N,\,r\leq s)=\\
%&=\int_{t_s^j}^{+\infty}\mathbb P(X^j_{v}-X^j_{t_s^j}\leq x\vert
%T_{s+1}^j=v,\mathbf {X}_{\mathbf{T}_r}={\bf
%x}_r,\mathbf{T}_r={\bf t}_r,\, r\in\mathbb N,\,r\leq s)\\
%&d\mathbb P(T_{s+1}^j\leq v\vert \mathbf {X}_{\mathbf{T}_r}={\bf
%x}_r,\mathbf{T}_r={\bf t}_r,\, r\in\mathbb N,\,r\leq s)=\\
&=\int_{t_s^j}^{+\infty}\mathbb P(X^j_{v}-X^j_{t^j_s}\leq x \vert
X^j_{t^j_r}=x^j_r,\,r\in\mathbb N,\,r\leq s
)\\
&d\mathbb P(T_{s+1}^j\leq v\vert\mathbf{T}_r={\bf t}_r,\,
r\in\mathbb N,\,r\leq s)
\end{aligned}$$
On the other hand
$$\begin{aligned}
   &\mathbb P(X^j_{T^j_{s+1}}-X^j_{T^j_s}\leq x \vert X^j_{T^j_r}=x^j_r,\mathbf{T}_r={\bf t}_r,\, r\in\mathbb N,\,r\leq s)=\\
%&=\int_{t^j_s}^{+\infty}
%\mathbb P(X^j_v-X^j_{t_s^j}\leq x\vert T^j_{s+1}=v,X^j_{T^j_r}=x^j_r,\mathbf{T}_r={\bf t}_r,\, r\in\mathbb N,\,r\leq s)\\
%&d\mathbb P(T_{s+1}^j\leq v\vert X^j_{T^j_r}=x^j_r,\mathbf{T}_r={\bf t}_r,\, r\in\mathbb N,\,r\leq s)=\\
&=\int_{t^j_s}^{+\infty} \mathbb P(X^j_v-X^j_{t_s^j}\leq x\vert
X^j_{t^j_r}=x^j_r,\,r\in\mathbb N,r\leq s) d\mathbb
P(T_{s+1}^j\leq v\vert \mathbf{T}_r={\bf t}_r,\,r\in\mathbb
N,r\leq s)
  \end{aligned}$$
and the thesis follows.
\end{proof}
\begin{remark}\label{reeem}
If we assume that the multidimensional process $\mathbf X$ has
Granger independent increments it is trivial to show that the
conclusion is
\begin{equation}\label{lemmaind}
\mathbb P(X^j_{T^j_{s+1}}-X^j_{T^j_s}\leq x\vert\mathcal
F_s^{X_T,T})=\mathbb P(X^j_{T^j_{s+1}}-X^j_{T^j_s}\leq x \vert
\mathcal F^{T}_s).
\end{equation}

\end{remark}\bigskip

\begin{proposition}\label{prop1}
Let $\mathbf S$ be a multivariate process such that each component
is a positive martingale with respect
 to its natural filtration and it is
not Granger caused by the others. Let $\mathbf T$ be a
multidimensional stochastic clock independent of $\mathbf S$.
Then, the time changed stochastic process $\mathbf {S}_{\mathbf
T}$ is a martingale with respect to its natural filtration.
\end{proposition}

\begin{proof}
Let $\mathbf X_t=(\ln (S^1_t),\ldots ,\ln (S^m_t))$.\\
By the hypotheses, for all $v\in\mathbb N$, $\mathbb E\left [\left
.e^{X^j_{v+1}-X^j_v}\right\vert\mathcal F^j_v\right ]=1$ and
thanks to the above Lemma, if $s\in\mathbb N$,
$$
\mathbb E\left [S^j_{T^j_{s+1}}-S^j_{T^j_s}\vert \mathcal
F^{X_T,T}_s\right ] =S^j_{T^j_s}\mathbb E\left
[e^{X^j_{T_{s+1}^j}-X^j_{T_s^j}}-1\vert \mathcal F^{X_T,T}_s\right
] =S^j_{T^j_s}\mathbb E\left
[e^{X^j_{T_{s+1}^j}-X^j_{T_s^j}}-1\vert \mathcal
F^{X^j_{T^j},T}_s\right ].$$ Since
$$\begin{aligned}
&\mathbb E\left [e^{X^j_{T_{s+1}^j}-X^j_{T_s^j}}-1\vert
X^j_{T^j_r}=x^j_r,\mathbf{T}_r={\bf t}_r,\, r\in\mathbb N,\,r\leq
s
\right ]=\\
%&=\int_{t_s^j}^{+\infty}\mathbb E\left
%[e^{X^j_{u}-X^j_{t_s^j}}-1\vert T^j_{s+1}=u, X^j_{T^j_r}=x^j_r,\mathbf{T}_r={\bf t}_r,\, r\in\mathbb N,\,r\leq s\right ]\\
%&d\mathbb P(T^j_{s+1}\leq u\vert X^j_{T^j_r}=x^j_r,\mathbf{T}_r={\bf t}_r,\, r\in\mathbb N,\,r\leq s)=\\
&=\int_{t_s^j}^{+\infty}\mathbb E\left
[e^{X^j_{u}-X^j_{t^j_s}}-1\vert X^j_{T^j_r}=x^j_r,\, r\in\mathbb
N,\,r\leq s \right ] d\mathbb P(T^j_{s+1}\leq
u\vert\mathbf{T}_r={\bf t}_r,\, r\in\mathbb N,\,r\leq s)=0,
\end{aligned}$$
each $S_{T^j}^j$ is a martingale with respect to the filtration
$\left (\mathcal F^{X_T,T}_t\right )$ and hence it is a martingale
with respect to the smaller filtration $\left (\mathcal
F^{X_T}_t\right )$ as required.

\end{proof}

Until now we have studied under which conditions the martingale
property is preserved through the time change. A natural question
now takes us to the second question. What kind of dependence will
be satisfied by the increments of the time changed process? We see
below that, given that $X$ has Granger independent increments,
adding the assumption of marginal stationary increments and of a
Granger independent clock yields the result and that under these
assumptions the time changed process exhibits Granger independent
increments.

\begin{proposition}\label{sab3}
Let $\mathbf X=(X^1,\ldots ,X^m)$ be a multidimensional stochastic
process with Granger independent and stationary marginal
increments and $\mathbf T$ be a multidimensional stochastic clock
with Granger-independent increments independent of $\mathbf X$.
Then the time changed process $\mathbf {X_T}$ has Granger
independent increments.
\end{proposition}

\begin{proof}
We start showing that, for $s\in\mathbb N$,
\begin{equation}\label{ciao}
 \mathbb P(X^j_{T^j_{s+1}}-X^j_{T^j_s}\leq x\vert \mathcal F^T_s)=\mathbb P(X^j_{T^j_{s+1}}-X^j_{T^j_s}\leq x)
\end{equation}
In fact,
$$\begin{aligned}
  &\mathbb P(X^j_{T^j_{s+1}}-X^j_{T^j_s}\leq x\vert \mathbf T_r={\bf t}_r,\,r\in\mathbb N,r\leq s)=\\
&=\int_0^{+\infty}\mathbb P(X^j_{t^j_s+u}-X^j_{t_s^j}\leq x\vert T^j_{s+1}-T^j_s=u, \mathbf T_r={\bf t}_r,\,r\in\mathbb N,r\leq s)\cdot\\
&\cdot d\mathbb P(T^j_{s+1}-T^j_s\leq u\vert\mathbf T_r={\bf t}_r,\,r\in\mathbb N,r\leq s)=\\
  &=\int_0^{+\infty}\mathbb P(X^j_{t_s^j+u}-X^j_{t_s^j}\leq x)
dF_{T^j_{s+1}-T^j_s}(u)=\int_0^{+\infty}\mathbb P(X^j_{u}\leq x)
dF_{T^j_{s+1}-T^j_s}(u),
  \end{aligned}$$
thanks to marginal stationarity and
$$\begin{aligned}
  &\mathbb P(X^j_{T^j_{s+1}}-X^j_{T^j_s}\leq x)=\\
&=\int_0^{+\infty}\int_0^{+\infty} \mathbb P(X^j_{v+u}-X^j_{v}\leq
x\vert T^j_{s+1}-T^j_s=u, T^j_s=v)
dF_{T^j_{s+1}-T^j_s, T^j_s}(u,v)=\\
&=\int_0^{+\infty}\int_0^{+\infty}\mathbb P(X^j_{v+u}-X^j_{v}\leq
x) dF_{T^j_{s+1}-T^j_s}(u)dF_{T^j_s}(v)
%&=\int_0^{+\infty}\int_0^{+\infty}\mathbb P(X^j_{u}\leq x)dF_{T^j_{s+1}-T^j_s}(u)dF_{T^j_s}(v)=\\
=\int_0^{+\infty}\mathbb P(X^j_{u}\leq x)dF_{T^j_{s+1}-T^j_s}(u)
  \end{aligned}$$
where $F_{T^j_{s+1}-T^j_s, T^j_s}(u,v)$ denotes the joint
cumulative distribution function of $(T^j_{s+1}-T^j_s, T^j_s)$ and
$F_{T^j_{s+1}-T^j_s}(u)$ the cumulative distribution function of $T^j_{s+1}-T^j_s$.\\
Hence (\ref{ciao}) is proved. \\
Now, by \text{ by (\ref{lemmaind})} and \text{ by (\ref{ciao})}
$$\begin{aligned}
\mathbb P(X^j_{T^j_{s+1}}-X^j_{T^j_s}\leq x\vert \mathcal
F^{X_T}_s) &=\mathbb E\left[\left .\mathbb
P(X^j_{T^j_{s+1}}-X^j_{T^j_s}\leq x\vert \mathcal
F^{X_T,T}_s)\right\vert\mathcal F^{X_T}_s\right ]
%&=\mathbb E\left[\left .\mathbb P(X^j_{T^j_{s+1}}-X^j_{T^j_s}\leq x\vert \mathcal F^T_s)\right\vert\mathcal F^{X_T}_s\right ]=\text{ by (\ref{ciao})}\\
%&=\mathbb E\left[\left .\mathbb P(X^j_{T^j_{s+1}}-X^j_{T^j_s}\leq x)\right\vert\mathcal F^{X_T}_s\right ]=\\
=\mathbb P(X^j_{T^j_{s+1}}-X^j_{T^j_s}\leq x)
\end{aligned}$$
and this trivially implies that $X_T$  has Granger independent
increments.
\end{proof}

\section{Conclusion}\label{conclusione}
In this paper we address the main flaw of copula function applications to the evaluation of multivariate equity derivatives, namely the fact that copula functions are intrinsically static objects that are used to link variables, rather than processes. As a result, in standard copula functions applications it is impossible to impose consistency relationships among prices of the same product for different maturities. The consistency required has to do with the martingale condition and the dependence structure among the levels. Particularly, the dependence structure of the levels at different time horizons must be made consistent with the dependence structure of the increments, which provides the characteristic feature of any multivariate stochastic process. \\
Here we propose and study a new class of processes for which
copula functions can be applied ensuring the intertemporal
consistency of prices. The requirements imposed to this class of
processes are that no process can be Granger-caused by any of the
others, and that each process is an univariate martingale with
independent increments: this class of processes is called GIMP
(Granger Independent Martingale Processes). The impact of a stochastic clock on the Granger causality assumption is analyzed and conditions under which this is preserved introduced.

\end{document}